\documentclass[11pt]{elsarticle}
\usepackage[utf8]{inputenc}
\usepackage{color}
\usepackage{amsmath,amssymb,latexsym}
\usepackage[vlined,english,boxed]{algorithm2e}
\usepackage[english]{babel}
\usepackage{amsmath} 
      \IncMargin{1em}  
      \SetKwInput{KwData}{Input}
      \SetKwInput{KwResult}{Output}
      \SetKwIF{If}{ElseIf}{Else}{If}{then}{else if}{else}{endif}
      \SetKwFor{ForEach}{For Each}{Do}{done}
      \SetKwFor{While}{While}{do}{done}
      \SetKwFor{Procedure}{Procedure}{}{}
      \SetKw{Return}{Return}%
      \SetKwComment{Comment}{/* }{ */}
      \SetSideCommentRight
      \DontPrintSemicolon

\usepackage{hyperref}
\usepackage{cleveref}      
\usepackage{microtype}

\newtheorem{definition}{Definition}
\newtheorem{theorem}{Theorem}
\newtheorem{lemma}{Lemme}
\newtheorem{corollary}{Corollary}

\newtheorem{proposition}{Proposition}

\newenvironment{proof}{\noindent\textit{Proof.}}{{}\hfill $\Box$\\}
\newcommand{\card}[1]{\left|{#1}\right|}

\newcommand{\set}[1]{\left\{{#1}\right\}}

\newcommand{\A}{\mathcal{A}}
\newcommand{\AT}{\mathcal{AT}}
\newcommand{\R}{\mathcal{R}}
\newcommand{\C}{\mathcal{C}}

\journal{}

\makeatletter
\def\ps@pprintTitle{%
	\let\@oddhead\@empty
	\let\@evenhead\@empty
	\def\@oddfoot{}%
	\let\@evenfoot\@oddfoot}
\makeatother

\begin{document}

\begin{frontmatter}
\title{Forwarding Tables Verification through\\ Representative Header Sets}

\author[liafa]{Yacine Boufkhad}
\author[inria]{Ricardo de la Paz}
\author[inria]{Leonardo Linguaglossa}
\author[nokia]{Fabien Mathieu}
\author[nokia]{Diego Perino}
\author[inria,liafa]{Laurent Viennot}

\address[liafa]{Université Paris Diderot - Paris 7}
\address[inria]{Inria}	
\address[nokia]{Nokia Bell Labs France}

\begin{abstract}
	Forwarding table verification consists in checking the distributed data-structure resulting from the forwarding tables of a network.  A classical concern is the detection of loops. We study this problem in the context of software-defined networking (SDN) where forwarding rules can be  arbitrary bitmasks (generalizing prefix matching) and where tables are updated by a centralized controller. Basic verification problems such as loop detection are NP-hard and most previous work solves them with heuristics or SAT solvers.
	
	We follow a different approach based on computing a representation of the header classes, i.e. the sets of headers that match the same rules.  This representation consists in a collection of representative header sets, at least one
	for each class, and can be computed centrally in time which is polynomial in the number of classes.  Classical verification tasks can then be
	trivially solved by checking each representative header set. 
	In general, the number
	of header classes can increase exponentially with header length, but it remains
	polynomial in the number of rules in the practical case where rules are
	constituted with predefined fields where exact, prefix matching 
	or range matching is applied in
	each field (e.g., IP/MAC addresses, TCP/UDP ports).  
	We propose general techniques that work in
	polynomial time as long as the number of classes of headers is polynomial and
	that do not make specific assumptions about the structure of the sets
	associated to rules. The efficiency of our method rely on the fact that
	the data-structure representing rules allows efficient computation of
	intersection, cardinal and inclusion.
	
	Finally, we propose an algorithm to maintain such representation in presence of updates (i.e., rule insert/update/removal). We also provide
	a local distributed algorithm for checking the absence of black-holes and
	a proof labeling scheme for locally checking the absence of loops.
\end{abstract}

\end{frontmatter}


%
%
%
%

\section{Introduction}

The diagnosis of network problems, such as the existence of loops (some packets may loop in the network) or black-holes (some nodes drops some packets that can be delivered elsewhere), is a
challenging issue. Indeed, routing decisions take into consideration several fields of packet headers, and the full header space cannot be exhaustively scanned. For example, the size of the IP destination field alone is $2^{32}$ for IPv4 and
$2^{128}$ for IPv6. The overall correctness of the forwarding tables of all devices of a network cannot
rely solely on the correctness of the protocols used to build them. The reasons are: different routing and management protocols interact simultaneously; manual configurations are operated by possibly several administrators.
Furthermore, the forwarding process of a large network is the result of the
interference of several device types (e.g., router, switches, middleboxes or
firewalls) and mechanisms (e.g., Ethernet, IP, NAT or MPLS) introduced in the
network to provide services including and beyond layer 2-3 forwarding.

The recent advent of Software Defined Networking (SDN)~\cite{Peresini:2013:OCP:2491185.2491205,Kotronis:2012:ORC:2390231.2390241,Monsanto:2013:CSN:2482626.2482629,Katta:2013:ICU:2491185.2491191} constitutes both an
opportunity and an issue for network problem diagnosis. The opportunity is that 
forwarding tables can be managed by a centralized SDN controller, where they can
be verified.  The issue is that SDN allows forwarding rules to be specified over multiple fields with arbitrary wildcard bitmasks (that generalize prefix
matching) covering an increasing number of protocol headers. 
Specifically, the predicate of
a rule $r$ is represented by a
string $m$ of $\ell$ letters in $\set{0,1,*}$ such that a header $h$ is in the
set when $m_i=*$ or $h_i=m_i$ for all bit position $i\in 1..\ell$. 
 As noted in
\cite{anteater} most verification tasks become NP-hard in this context.
Consider
for example a task as simple as testing whether some header may fail to match
any rule in a forwarding table.  One can easily see that it is equivalent to a
SAT problem with variables $h_1,\ldots,h_\ell$. 

In this paper, we try to provide some fundamental ground to the approach
proposed by Khurshid et al.~\cite{veriflow} where a representation of equivalent
classes of headers is maintained. Each verification task is then performed on each class instead of each possible header. Classes can be defined based on
identical network behavior or simply on identical set of rules matched.
The advantage of the latter definition is that it is order invariant 
(changing the priority of rules in a table does not affect the classes) and is thus
more stable with respect to minor changes. Indeed, Khurshid et al. implicitly
rely on this finer notion of class for tractability reasons.
The drawback with respect to the former definition is that the number of classes
may be larger. 
However, when we consider decision rules that appear in practice using
prefix matching or range matching on a fixed number $d$ of fields, the
number $c$ of classes for $n$ rules is $O(n^d)$ and remains polynomial.
We indeed believe that this bound is not tight and that verification tasks
appeared tractable in implementations such as 
\cite{anteater,libra, veriflow, netplumber}
because $c$ remains relatively small in practice.  
Not surprisingly, $c$ can be exponential in $n$ in general.
We thus propose to re-investigate the problem of forwarding table verification
using the number $c$ of header classes as main complexity parameter.

Our work is inspired by seminal work of Boutier and Chroboczek~\cite{boutier}
who explore the problem of disambiguating a forwarding table. Ambiguity may
arise in practical implementations of source sensitive routing when
two decision rules intersect without one having clearly precedence over the other.
They note that a more specific rule  always have precedence in practical implementation of IP stack (a rule that applies on set $r$ is more specific than a rule that applies on a set $r'$ when $r\subseteq r'$). They thus propose to
add more specific forwarding rules in case of conflicts to eliminate any
ambiguity (typically a rule for $r\cap r'$ is inserted when $r$ and $r'$
are in conflict). They then explicit a sufficient and necessary condition
for a set of rules to be ambiguity free that they call weak completion.
However, they use a stronger condition called strong completion obtained
by adding all possible intersections of rules. They propose 
algorithms for incrementally
maintaining strong completion of the set of rules in the case where
rules consist in prefix matching on two fields.
This can be considered as a first step for representing classes of headers
since their work allows to represent intersection of rules. However, a 
further step is needed as classes are intersections of rules and 
complementary of rules.

\paragraph{Related work: }
Most previous work compute sub-classes of headers by means of combining
rules by intersection and set difference.
Veriflow, the solution presented by Khurshid et al.~\cite{veriflow}
store these combinations in a multi-dimensional
trie assuming that rules are based on multi-field matching
with wildcard masks. The intersection of two wildcard 
masks is either empty or can be represented with a single mask. However,
the difference of two sets cannot be represented with a single mask in general.
The size of their representation may thus grow unexpectedly due
to set differences. However, in practice, masks are often prefixes
and correspond to a range of natural numbers. In that case, 
the difference of two
prefixes is expressed as two ranges at most. The efficiency
of the solution certainly benefit from this optimization.
HSA~\cite{hsa} and NetPlumber~\cite{netplumber} rely
similarly on representation of header space through combinations of rules with
intersection and set difference. Interestingly,
some bound guaranty on representation size is derived in \cite{netplumber}
when the network satisfies  a property called linear fragmentation.
Anteater~\cite{anteater} represents verification tasks 
as boolean satisfiability problems and solve them using a SAT solver. 
The boolean formulas
generated by this approach equivalently represent combinations of rules
with intersection and set difference.
Libra~\cite{libra} assumes that rules are based on prefix matching and
is concerned by solving the verification tasks in a high performance perspective
using MapReduce framework. One can show that the number $c$ of header classes
is linear in the number of rules when they all are prefixes. The context
of prefix matching rules is thus significantly simpler.

In the context of forwarding table disambiguation introduced by
Boutier and Chroboczek~\cite{boutier}, the problem is to cover conflicts
between rules. The set of headers were two rules are in conflict is indeed
their intersection. Considering combinations of rules with intersection
solely is thus sufficient for covering all conflicts. This context is thus
simpler than the context of forwarding table verification were a header
class is indeed the intersection of the rules matched by headers of the class
minus the rules not matched. 

\paragraph{Our contribution: }
Our main contribution is to show that an exact representation of the
$c$ classes generated by $n$ rules can be computed in polynomial time
in $c$. Surprisingly, this can be done without computing any set difference,
solely using intersections, inclusion tests and cardinal computations.
For that purpose, we
show that the notion of weak completion introduced by Boutier and
Chroboczek~\cite{boutier} is tightly related to that of header classes and is a
basic brick for constructing a representation of the classes generated
by a set of rules in $O(n\ell c^2)$ time where $\ell$ is the header
bit length.
This representation is optimal with respect to space as it consists in
exactly $c$ representative header sets. Each representative header set $s$
is represented similarly as a rule. The rules matching the headers of the
class associated to $s$ can be identified as those containing $s$.
Interestingly, this representation thus allows to efficiently compute
the routing decisions taken for all headers in the class without
having to find a representative header  
in the class, a task that is NP-hard in general.
Our technique do not require any specific data-structure for
representing rules and
header sets, it just relies on the fact this representation has size
$O(\ell)$ and that intersection, inclusion, and cardinal can be computed
in $O(\ell)$ time.

Once such a representation is computed all classical verification tasks
can be performed in $O(f\ell c)$ time where $f$ is the sum of table sizes. 
We thus obtain that
verification of forwarding tables can be performed in polynomial time
with respect to the number $c$ of header classes that should be tested.
Note that this does not contradict the difficulty of verification tasks 
with respect to the size of the input
since $c$ can be exponential in the number $n$ of rules in general.
Our algorithms are indeed incremental and allow to maintain 
such a representation for a collection $\R$ of rules under insertion 
and deletion of rules in $\R$. Each operation can be performed in 
$O(\ell c^2)$ time. As we make minimal assumptions about the data-structure
used for representing header sets, we did not try to optimize these algorithms.

Stepping away from SDN, we show that in a distributed
context a basic verification task such as absence of black-holes can be 
performed locally with our techniques. We show how our representation
allows to produce a proof labeling scheme allowing to locally check
the absence of loops. Distributed verification of absence of loops is
explored when the network uses more specific routing, i.e. the rules
matching a packet header become more specific along the route, 
an assumption which is natural in hierarchical networks.

The rest of the paper is organized as follows.
\Cref{sec:problem-statement} presents our network model the problem of
forwarding table verification. \Cref{sec:preliminaries} introduces basic
tools concerning algebra of sets built upon the notion of weak completeness.  In
\Cref{sec:representative-header-sets}, we derive our main result
concerning the computation of an optimal representation of header classes
through a process that we call weak completion.
\Cref{sec:distributed-verification} describes how main verification tasks
(absence of black-holes and loops) can be performed in a distributed
manner. Finally, \Cref{sec:conclusion} concludes the paper.

\section{Problem statement}\label{sec:problem-statement}  

\subsection{Network model}

We consider a general model of network where packets are forwarded according to
the content of their fixed length header. 
A \emph{network instance} $\mathcal{N}$  is given by
the header bit-length $\ell$, a graph $ G=(V,E) $ 
and the \emph{forwarding table} $T(u)$ of each node $u\in V$. 
A node $u\in V$ is also called a \emph{router}. Each forwarding table
$T(u)$ is an ordered list of forwarding rules $r_1,\ldots,r_k$.  Each forwarding rule, or
simply \emph{rule}, is made of a predicate and an action to apply on any packet whose header matches the predicate. For ease of notation, we also call rule the set of headers that match a given rule, so we say that a header $h$ \emph{matches} a rule $r$ when $h\in r$.
We consider three possible actions: forward the packet to a neighbor, drop the packet, or deliver the
packet (when the packet is arrived at destination).
The priority of rules is given by the order of the rules: when a packet with header $h$ arrives at node
$u$, the first rule matched by $h$ is applied. Equivalently, the rule $r_i$ is applied when $h\in r_i\cap \overline{r_1} \cap \cdots \cap
\overline{r_{i-1}}$, where $\overline{r}$ denotes the complement of $r$.  When
no match is found (i.e. $h\in \overline{r_1} \cap \cdots \cap \overline{r_k}$),
the packet is dropped. We say that $u$ takes action \emph{$v$-forward} 
(respectively \emph{drop}, or \emph{deliver}) on header $h$ when the first rule matching $h$ in $T(u)$ indicates 
to forward the packet to $v$ (respectively to drop it, or to deliver it).
The collection $\R=\cup_{u\in V}T(u)$ of the sets modeling any rule in
the network is called the \emph{rule collection} of $\mathcal{N}$.

In practice forwarding rules are represented through simple data structures
such as a wildcard mask or a range of integers when considering the bits of
the header as the binary representation of an integer.
In both cases the representation of the intersection $r \cap r'$ of two rules
$r$ and $r'$ can be efficiently computed and represented within the same type of
data structure. Similarly, we can test efficiently whether $r\subseteq r'$ 
and compute easily the number $\card{r}$ of headers matching a rule $r$ in 
both cases.
This remarks also apply in the context of SDN or firewalls, where 
multi-field rules are considered. In SDN, a rule is the Cartesian product 
of several wildcard masks. In firewalls, rules are typically expressed as
Cartesian products of ranges.
In the sequel, we 
assume that the data structure  representing a rule uses space $O(\ell)$, 
and that intersection, inclusion and cardinal
can be computed in $O(\ell)$ time. 

Let $H$ denote the set of all headers with $\ell$ bits.  
We say that two headers $h$ and $h'$ are
\emph{rule equivalent} for a collection $\R$ of rules if they match the same
rules in $\R$ and define the \emph{header classes} of the collection $\R$ as the
equivalence classes of this relation. Two rule equivalent headers obviously
follow the same forwarding decisions in all routers of any network $\mathcal{N}$
with rule collection $\R$. 
We can thus define a directed graph $G_c$ for each header
class of $\R$ with vertex set $V$.  Each node in $G_c$
is labeled as drop, forward or
deliver according to the action it takes on headers in $c$.  The arcs of $G_c$
correspond to forward actions: $uv\in E(G_h)$ when $u$ takes action $v$-forward
on headers in $c$.

\subsection{Verification tasks}

Forwarding table verification consists in verifying that some
network wide properties are satisfied. Such properties are 
classically among:

\begin{itemize}
\setlength{\itemsep}{0pt}%
\item NO-LOOP: no packet can loop in $G$, i.e. there is no header class
  $c$ such that $G_c$ has a loop.
\item NO-BLACKHOLE: no packet can be dropped by a node when it is delivered
  or forwarded by another node: for all header class $c$, 
  all nodes in $G_c$ have
  label drop, or none of them.
\item REACHABILITY$(u,v)$: for fixed nodes $u,v$ in $V$, some packet can 
  travel from $u$ to $v$, i.e. $G_c$ has a path from $u$ to $v$ for some
  header class $c$.
\item CONSISTENCY$(u,v)$: for fixed nodes $u,v$ in $V$, any packet has
  same fate when initiated at $u$ or $v$.
\end{itemize}

Most prominently, NO-LOOP and NO-BLACKHOLE together imply that the forwarding
tables of a network $G$ are basically correct: NO-LOOP implies that for any
header class $c$, $G_c$ is a
forest rooted at nodes with non-forward label. Additionally, NO-BLACKHOLE
ensures that for any header $h$ that can be delivered at some node in $G$, the
graph $G_c$ of its header class $c$ 
do not contain any drop label, it is thus a forest rooted at
nodes that deliver $h$. (Note that in complex hierarchical networks, some
destinations can be reached through several nodes: consider a multi-homed
sub-network for example.)

All the verification tasks we are aware of can be stated either as 
$\forall c, P_c$ or $\exists c, P_c$ where $P_c$ is some basic property
that can be tested in linear time on $G_c$.

\section{Preliminaries}\label{sec:preliminaries}

As we model rules as sets of headers, we will interchangeably use the terms
rule and set. We use the term rule to emphasize that the considered set
is associated to a rule in a forwarding table of some network. 
All considered sets are subset of $H$, the set of all possible headers
(headers may be simply called elements).
We use the term \emph{collection} for a set of sets.
As any boolean combination of decisions taken by the routers of a network
can be expressed in terms of set-algebraic operations, we
review some basic facts about algebras of sets.

\subsection{Algebra of sets}

Given a collection $\R$, we let $\A(\R)$ denote the algebra of sets
generated by $\R$, that is the minimal collection including $\R$ that
is closed under intersection, union and complement. An \emph{atom} is any 
non-empty element of $\A(\R)$ which is minimal for inclusion.
The atoms of $\A(\R)$ form a partition of $H$ and $\A(\R)$ is isomorphic
to the power set of its atom collection. This classical result
can be highlighted by the following proposition.

\begin{proposition}[Folklore]
\label{prop:atoms}
The header classes of a collection $\R=r_1,\ldots,r_n$ are the atoms of
$\A(\R)$. In other words, each atom $A$ of $\A(\R)$ is equal to
$r'_1\cap\cdots\cap r'_n$ for some sets $r'_1,\ldots,r'_n$ such that each $r'_i$
is either $r_i$ or $\overline{r_i}$.
\end{proposition}
We include a proof for the ease of readers unfamiliar with algebra of sets.

\begin{proof}
Given an element $h\in H$ let $r'_i=r_i$ if $h\in r_i$ and $r'_i=\overline{r_i}$
otherwise. The header class of $h$ is thus
$r'_i\cap\cdots\cap r'_n$ (the set of headers matching the same rules as $h$)
which is an element of $\A(\R)$.
Let $\C$ denote the collection of any finite union of such classes
(including the empty union $\emptyset$). 
As $\C$ is clearly closed by intersection, union and complement and $\C$
includes $\R$ (any rule $r_i$ can be written as the union of classes of headers
matching $r_i$), we indeed have $\C=\A(\R)$ by minimality of $\R$.
Any atom is thus an element of $\C$ and can be written as a union of classes.
Each atom must indeed be a class by minimality of atoms.
Each class is minimal as it intersects no other atom than itself and is
thus an atom.
(Note that $\C=\A(\R)$ is isomorphic to the power set of the set of atoms.)
\end{proof}
The header classes of a collection $\R$ will thus
be simply called the \emph{atoms} of $\R$ in the sequel.

\subsection{Weak completeness}

The following definition is borrowed from Boutier and Chroboczek~\cite{boutier}.
We have slightly modified it with an additional requirement concerning
the set $H$ of all elements.

\begin{definition}[\cite{boutier}]
A collection $\R$ is 
\emph{strongly complete} (respectively \emph{weakly complete})
iff $H\in \R$ (respectively $H=\cup_{r\in \R}r$)
and for any sets $r,r'\in \R$, we have $r\cap r'\in \R$
(respectively $r\cap r' = \cup_{r''\subseteq r\cap r'} r''$).
\end{definition}

Obviously, a strongly complete collection is also weakly complete.
We will see that the notion of weak
completeness is appropriate for identifying the atoms of any collection.
We first identify the atoms of a weakly complete collection.

\begin{proposition}
\label{prop:weak}
The atoms of any weakly complete collection $\R$
are the sets $a(r)=r \setminus \cup_{r'\subsetneq r}r'$ for
$r \in \R$ such that $a(r)$ is not empty.
A set $r\in \R$ such that $a(r)\not=\emptyset$ is said to be
\emph{atom representative}. We let $\AT(\R)$ denote the collection of
the atom representative sets of $\R$.
\end{proposition}

\begin{proof}
We first note that two headers $h,h'$ in $a(r)$ 
for some set $r\in\R$ must belong to the same sets $r'\in\R$ implying that
$a(r)$ is included in some atom of $\R$ according to 
\Cref{prop:atoms}
and is thus an atom as  $a(r)\in\A(\R)$. 
For the sake of contradiction, suppose that
some $r'\in\R$ contains $h'$ and not $h$. Weak completeness implies
$r\cap r'=\cup_{r''\subseteq r\cap r'}r''$ and $h'$ must be in some set $r''\in\R$
such that $r''\subseteq r\cap r'$. As $r''$ does not contain $h$, it
is strictly included in $r$ and it is disjoint from $a(r)$ by definition 
of $a(r)$. This then contradicts $h'\in a(r)\cap r''$.

We now prove that any atom $A$ of $\R$ is equal to $a(r)$ for some $r\in\R$.
Consider $h\in A$.
By weak completeness, we have $H=\cup_{r\in \R}r$ and $h$ must be in some 
$r\in \R$. Consider such a set $r\in \R$ containing $h$ which is minimal for
inclusion. If $h$ is not in $a(r)$, it must be in some $r'\subsetneq r$
according to the definition of $a(r)$. This then contradicts the 
minimality of $r$. We thus have $h\in a(r)$ and $a(r)=A$ since $a(r)$ is
an atom.

To conclude, we finally show 
that the sets $a(r)$ for $r\in \R$ are pairwise disjoint, implying
that two non empty such sets $a(r)$ and $a(r')$ must be different for
$r\not= r'$.  For
the sake of contradiction, suppose $a(r)\cap a(r')\not=\emptyset$ for
$r\not= r'$. If $r\subsetneq r'$ then $r$ is disjoint of $a(r')$ by definition
and so is $a(r)$ which is included in $r$. 
Now consider the case $r\cap r'\subsetneq r$.
From the weakly complete property, we have 
$r\cap r'= \cup_{r''\subseteq r\cap r'}r''$. As $r''\subseteq r\cap r'$
implies $r''\subsetneq r$ and $a(r) =  r \setminus \cup_{r''\subsetneq r}r''$,
we can conclude that $a(r)$ cannot contain any element in $r\cap r'$ and it is
thus disjoint from $a(r')$ since $a(r')\subseteq r'$.
\end{proof}
The following proposition shows that the atom representative sets
of a weakly complete collection can be identified efficiently.

\begin{proposition}
\label{prop:card}
Given a weakly complete collection $\R$, the cardinals of the
sets $a(r)=r \setminus \cup_{r'\subsetneq r}r'$ for $r \in \R$
can be computed in time $O(|\R|^2\ell)$.
\end{proposition}

\begin{proof}
We first prove $r= \cup_{r'\subseteq r} a(r')$ for all $r\in \R$.
As $a(r')\subseteq r'$ for each $r'$, it is sufficient to prove 
$r \subseteq \cup_{r'\subseteq r} a(r')$. Consider $h\in r$. If no rule
$r'\subsetneq r$ contains $h$, then $h$ is in $a(r)$ by definition.
Otherwise, consider a rule $r'\subsetneq r$ containing $h$ which is minimal
for inclusion. Then $h\in a(r')$ since no rule $r''\subsetneq r'$ can contain
$h$ by minimality of $r'$.
We can thus write 
$a(r)=r \setminus \cup_{r'\subsetneq r} \cup_{r''\subseteq r'} a(r'')$
or equivalently $a(r)=r \setminus \cup_{r''\subsetneq r} a(r'')$.
We can thus compute the cardinals $\card{a(r)}$ for $r \in \R$
in a dynamic programming fashion: first compute the inclusion relation in
time $O(|\R|^2\ell)$. Then consider the rules in $\R \cup \set{H}$
according to a topological order. The cardinal of each $a(r)$ can then
be computed from the cardinals of $r$ and $a(r')$ for $r'\subsetneq r$
as $\card{a(r)} = \card{r} - \sum_{r'\subsetneq r} \card{a(r')}$. Each
cardinal computation takes time $O(|\R|\ell)$.
\end{proof}

\subsection{Covering a collection with another}

As we will try to represent the atoms of a collection $\R$ with another
collection $\C$, we introduce the following notion of covering.
We say that a collection $\C$ \emph{covers} 
a set $r$ when $r=\cup_{r'\subseteq r, r'\in\C}\ r'$. Note that $\C$ clearly
covers $r$ when $r\in \C$. We say that $\C$ covers another collection $\R$
when it covers every set in $\R$. The following lemma shows how the 
two collections are then related in terms of atoms.

\begin{lemma}
\label{lem:repr}
Given a collection $\C$ covering a collection $\R$,
the atoms of $\C$ \emph{refine} those of $\R$ in the sense that
each atom of $\R$ is the union of some atoms of $\C$.
\end{lemma}

\Cref{lem:repr} follows easily by defining an atom as intersection
of sets or complement of sets from $\R$ and applying the covering property.
We now note that the notion of atom representative set is related to that
of covering.

\begin{lemma}
\label{lem:atomcov}
Any weakly complete collection $\C$ is covered by the collection $\AT(\C)$
of its atom representative sets.
\end{lemma}

Before proving this lemma, note the two following corollaries that can be derived from it.

\begin{corollary}
\label{cor:atomcov}
If a collection $\R$ is covered by
a weakly complete collection $\C$, it is also clearly covered by the atom
representative sets of $\C$.
\end{corollary}

\begin{corollary}
\label{cor:atomweak}
Given a weakly complete collection $\C$, the collection $\AT(\C)$ of its atom
representative sets is also weakly complete.
\end{corollary}

\begin{proof}[of \Cref{lem:atomcov}]
Suppose for the sake of contradiction
that $\C$ is not covered by $\AT(\C)$ and consider a minimal set $r\in\C$
that is not covered by $\AT(\C)$. As $\AT(\C)$ contains 
all atom representative sets
and obviously cover them, we thus infer that $r$ is not atom representative
and satisfies $r=\cup_{r'\subsetneq r, r'\in \C}\ r'$.
The minimality of $r$ implies that each $r'\subsetneq r$ is covered by
$\AT(\C)$, i.e. $r'=\cup_{r''\in\AT(\C), r''\subseteq r'}\ r''$.
We thus get $r\subseteq\cup_{r''\in\AT(\C), r''\subseteq r}\ r''$. As a union
of sets included in $r$ is obviously included in $r$, we indeed have
$r=\cup_{r''\in\AT(\C), r''\subseteq r}\ r''$ which contradicts the fact that $r$
is not covered by $\AT(\C)$.
\end{proof}

\section{Representative header sets}\label{sec:representative-header-sets}

A straightforward way to verify that some property is satisfied by a network $G$
for all possible headers would be to identify one header at least in each
atom of the collection of rules of $G$ and then to test the property
for each such representative header. However, computing such headers
(even a single one) can be NP-hard as stated in \cite{anteater}. 
We rather propose to consider representative header sets defined as follows.

\begin{definition}
A \emph{representative set} $s$ for an atom $A$ of a collection $\R$
is a set $s$ containing $A$ such that a set $r\in \R$ contains $A$
iff it contains $s$. A \emph{representative collection} for $\R$
is a collection of representative sets, at least one for each atom of $\R$.
A representative collection is optimal if it has exactly one representative
set for each atom of $\R$.
\end{definition}

The idea is that a representative set $s$ for an atom $A$
allows to test efficiently what are the rules $r$ matched by any $h\in A$
by testing $s\subseteq r$ instead of $h\in r$ for each $r\in \R$.
Obviously, an atom $A$ is a representative set for itself.
Note that for any set $s\supseteq A$, any set $r\in \R$ containing $s$
contains also $A$. On the other hand, if $s$ is too large a set $r$ containing 
$A$ may not contain $s$. The following proposition shows how a weakly complete
collection can provide a representative collection.

\begin{proposition}
\label{prop:weakrepr}
If a weakly complete collection $\C$ covers a collection $\R$,
then the collection $\AT(\C)$ of its atom representative sets 
is a representative collection for $\R$.
\end{proposition}

\begin{proof}
Consider an atom $A$ of $\R$. According to \Cref{lem:repr}, $A$ must
contain some atom $A'$ of $\C$ which is associated to some atom representative
set $s\in\AT(\C)$ such that $A'=s\setminus\cup_{s'\subsetneq s}s'$.  
We show that $s$ is
a representative set for $A$.  Any rule $r\in\R$ containing $s$ contains some
headers of $A$ and thus contains $A$.  Now, consider some set $r\in\R$
which contains $A$. 
As $r$ is covered by $\C$,  we have $r=\cup_{c\subseteq r, c\in \C}\ c$. 
The weak complete property then imply
$s\cap r=\cup_{c\subseteq s\cap r, c\in \C}\ c$. 
As $s\cap r$ contains $A'$, 
some set $c\subseteq s\cap r$ must contain some element of $A'$.
As $A'$ is an atom of $\C$, $c$ contains $A'$. 
As $A'=s\setminus\cup_{s'\subsetneq s}s'$, $c$ cannot be strictly included in $s$.
We thus have $c=s$ and $r$ contains $s$.
\end{proof}

\subsection{Partial completion}

A natural approach for approximating the atoms of $\R$ is to compute
combinations of sets in $\R$ by intersection. 
We call \emph{combination} of $\R$ any
non-empty set obtained as the intersection of some sets in $\R$
(possibly, the set $H$ of all elements is obtained as an empty intersection).
A collection $\C$ is called a \emph{partial completion} of $\R$ if it
is composed of combinations of $\R$. 
The following lemma is somehow symmetrical to \Cref{lem:repr}.
It is a direct consequence of \Cref{prop:atoms}.

\begin{lemma}
\label{lem:completion}
Given any partial completion $\C$ of a collection $\R$,
the atoms of $\R$ refine the atoms of $\C$,
i.e. each atom of $\C$ is the union of some atoms of $\R$.
\end{lemma}

\subsection{Weak completion}
\label{sub:weakcompl}

The number of all possible 
combinations can be much larger than the number of atoms.
For example, the number of combinations of $n$ ranges can be as large as
$n(n-1)/2$ (consider for example $[1,n+1], [2,n+2],\ldots, [n,2n]$)
when the number of their atoms is always $2n+1$ at most.
We thus focus on partial completions having sufficient properties
for representing the rule collection we are interested in.
The following lemma shows how weak completeness can help.
Its proof directly follows from the weak completeness definition and is omitted.

\begin{lemma}
\label{lem:cover}
Any weakly complete collection $\C$ 
covering a collection $\R$ covers any combination of $\R$. 
\end{lemma}

We define a \emph{weak completion} of $\R$ as
a partial completion $\C$ of $\R$ that
is weakly complete and covers every rule in $\R$.
Note that \Cref{lem:repr,lem:completion} imply that the atoms
of a weak completion of $\R$ are the same as those of $\R$.
\Cref{prop:weakrepr} then implies that $\AT(\C)$ is an optimal
representative collection for $r$.
We thus get the following corollary.

\begin{corollary}
\label{cor:wcompl}
Given any weak completion $\C$ of a collection $\R$,
the collection $\AT(\C)$ of atom representative sets in $\C$
is an optimal representative collection for $\R$.
\end{corollary}

The following algorithm shows how to compute a weak completion of $\R$ which
is minimal. Note that such minimal weak completion is indeed unique.

\begin{algorithm}[htb]
  $\C := \set{H}$ \;
  \ForEach{$r \in \R $}{
    \ForEach{$c \in \C $}{
      \lIf{$c\cap r\notin \C$}{add $c\cap r$ to $\C$.} \;
    }
    Compute $|a(c)|$ for all $c\in \C$. \;
    Remove from $\C$ any $c$ such that $|a(c)|=0$. \;
  }
  \caption{Computing a weak completion of a rule collection $\R$.}
  \label{algo:weak}
\end{algorithm}

\begin{theorem}
\label{th:repr}
Given any rule collection $\R$, \Cref{algo:weak} computes
the minimal weak completion of $\R$ in $O(|\R|c^2\ell)$ time where $c$ is the
number of atoms of $\R$.
\end{theorem}

The proof basically relies on \Cref{cor:atomweak,cor:atomcov} and on the fact that weak completeness of $\C$ is
maintained when adding a new set $r$ to it if we also add all $c\cap r$ for
$c\in \C$. The uniqueness of the minimal weak completion follows easily from
\Cref{lem:cover} and \Cref{prop:weakrepr}. The computation time
of \Cref{algo:weak} is dominated by the computation of the cardinals of
$a(r)$ for $r\in\C$ at each iteration of the main for loop and its time
complexity is a consequence of \Cref{prop:card}.
The details of the proof are similar to previous proofs and is deferred
to Appendix~A.

As the weak completion computed by  \Cref{algo:weak} 
contains only atom representative sets by construction,
\Cref{cor:wcompl} implies that it is an optimal representative
collection for $\R$. We thus get the following corollary.

\begin{corollary}
\label{cor:repr}
Given any rule collection $\R$, an optimal representative
collection for $\R$ can be computed in $O(|\R|c^2\ell)$ time where $c$ is the
number of atoms of $\R$.
\end{corollary}

Given a network $\mathcal{N}$ with rule collection $\R$ having $c$ atoms,
any classical verification task can obviously be centrally checked in
$O(c\ell\sum_{u\in V}\card{T(u)})$ time once an optimal representative collection
for $\R$ has been computed.
The reason is that the action taken by a router $u\in V$ with 
forwarding table $T(u)$ on any header of an atom $A$ 
is the action associated to first
rule of $T(u)$ containing $s$ where $s$ is the atom representative of $A$
(or drop if no rule contains $s$)
and can be identified in $O(\card{T(u)}\ell)$. Classical verification tasks
including NO-LOOP, NO-BLACKHOLE, 
REACHABILITY$(u,v)$, CONSIS\-TENCY$(u,v)$
can then be performed in linear time on the graph resulting from
actions of all routers.

Note that \Cref{algo:weak} relies on adding incrementally rules.
The following algorithm indicates how to maintain a weak completion
$\C$ of a collection $\R$ when a set $r$ is removed from $\R$.
It relies on the fact that $\C$ is a partial combination of $\R$
and assumes that we maintain for each combination $c\in\C$ the list
$R(c)$ of all sets in $\R$ that contain $ c $, and for each set $r\in\R$
the list $C(r)$ of all combinations $c$ such that $R(c)$ contains $r$.

\begin{algorithm}[htb]
  \ForEach{$c\in C(r)$}{
    Remove $r$ from $R(c)$ and replace $c$ by $\cap_{r'\in R(c)}r'$.\;
  }
  Compute $|a(c)|$ for all $c\in \C$. \;
  Remove from $\C$ any $c$ such that $|a(c)|=0$. \;
  \caption{Incremental deletion of $r\in\R$ in a weak completion $\C$ of $\R$.}
  \label{algo:del}
\end{algorithm}

The proof that the collection $\C'$ obtained by \Cref{algo:del}
is a weak completion of $\R'=\R\setminus\set{r}$ is similar to other proofs
of the paper and is omitted.

\section{Distributed verification}\label{sec:distributed-verification}

In the context of distributed verification, we focus on 
NO-BLACKHOLE and NO-LOOP tasks.

\subsection{Locally checking NO-BLACKHOLE}

NO-BLACKHOLE can easily be checked locally 
when $G$ is connected. 
Suppose a node 
$u\in V$  drops $h\in H$ when some node $v\in V$ forwards 
it or delivers it. Consider a path $P$ from $u$ to $v$ in the graph $G$.
Then there must exist two consecutive nodes $u',v'$ on $P$ such that
$u'$ drops $h$ and $v'$ forwards it or delivers it.
It is thus sufficient to perform the following 1-round computation
to check NO-BLACKHOLE. Each node $u$ sends its
forwarding table to its neighbors and symmetrically receive the table
of each neighbor.
Node $u$ then performs the following test for each neighbor $v$.
Compute a representative collection $\C$ of $T(u)\cup T(v)$.
For any rule $r_j$ of $T(v)$ with action deliver or forward,
check that any header $h$ following this rule
will follow a rule $r'$ of $T(u)$ with action
forward or deliver. 
As $\C$ is a representative collection of $T(u)\cup T(v)$,
it is sufficient to test that for each $c\in \C$ if the first rule 
containing $c$ in $T(v)$ has action deliver or forward, then so does
the first rule containing $c$ in $T(u)$ (the test fails if no rule in 
$T(u)$ contains $c$). The local computation is thus polynomial
in the maximum number of atoms of $T(u)\cup T(v)$ for $v\in N(u)$.

\subsection{Proof labeling scheme for NO-LOOP}

As loop detection cannot be performed locally, we propose a proof
labeling scheme for NO-LOOP (see e.g. \cite{prooflab} for 
a formal definition and an overview of proof labeling schemes). 
Our scheme is inspired by
the idea of labeling each node of a tree with its distance to the root,
a classical labeling scheme allowing
to check locally that the tree is indeed a tree (such a
scheme was introduced in \cite{labtree}).

Our scheme is based on local representative header sets defined as follows. 
For every node $u$ in $G$, we label $u$ with a weakly complete representative
collection $L(u)$ of $T(u)$ where each $c\in L(u)$ is associated with a
distance estimation $D(u,c)$.
We require that for any header $h$ that can follow a path $u_1,\ldots,u_k$ 
by applying rules $r_1,\ldots,r_k$, we have $D(u_1,c_1)>\cdots >D(u_k,c_k)$
for some sets $c_1,\ldots,c_k$ such that
for all $i\in 1..k$, $c_i\in L(u_i)$, $h\in c$ and $c\subseteq r_i$.
Given such labels, 
testing that distance labels always decrease during forward actions
can be checked locally in a similar manner as for NO-BLACKHOLE.
The network obviously satisfies NO-LOOP when the test succeeds.

Such a labeling may computed from a weak completion $\C$
of the rule collection $\R$ of
any network $\mathcal{N}$ satisfying NO-LOOP. Simply set $L(u)$ to $\C$ for all
$u\in V$. (Such a collection $\C$ is weakly complete and is a representative
collection for $\R$ according to \Cref{prop:weakrepr}.)
NO-LOOP then ensures that for each $c\in \C$ 
appropriate distance labels $D(u,c)$ can be associated to $c$ in each $L(u)$.
This labeling scheme can be very inefficient in terms of space as 
$\C$ can be exponentially larger than $\R$ which itself can be much larger
that any forwarding table of the network. We differ to future work the
study of how to gain space by using locality in labels.
However, we show next how some practical assumption allows to switch to a
very simple labeling scheme.

\subsection{More specific routing}

Interestingly, hierarchical networks use more specific routing in the
sense that rules applied successively to a packet during routing become
more and more specific.
Formally, we say that a rule $r'$ is
\emph{more specific} than a rule $r$ when $r'\subseteq r$.
A network satisfies MORE-SPECIFIC when the two following properties are
satisfied: a more specific rule always
has precedence (as with longest prefix matching for IP routing for example),
and the rules applied successively during routing of a packet become
more and more specific. The idea is that general rules are applied 
on nodes far away from the destination and rules become more specific
as a packet reach nodes closer to the destination. In hierarchical networks,
nodes in the subnetwork of the destination have finer grain view than
nodes outside the subnetwork. 
Checking that a network satisfies MORE-SPECIFIC can
clearly be checked locally similarly as NO-BLACKHOLE.

Now NO-LOOP can be tested with a simpler proof labeling scheme when
MORE-SPECIFIC is satisfied. If a packet follows a loop by successive
application of rules $r_1,\ldots,r_k$ in $u_1,\ldots,u_k$, we must
have $r_1=\cdots=r_k$ as MORE-SPECIFIC imply $r_1\supseteq \cdots\supseteq r_k$
and $r_k\supseteq r_1$. Distance labels in each node for all the rules $r$
in the network with same header set $r=r_1$ can be computed distributively with
Bellman-Ford algorithm (this is the basis of distance vector routing protocols).
In that case, fully distributed verification of NO-LOOP is thus possible.

\section{Conclusion}\label{sec:conclusion}

We have shown the tight fundamental link between weak completion as introduced
by Boutier and Chroboczek~\cite{boutier} and header equivalent classes 
as defined in Veriflow~\cite{veriflow}.
We have shown how to compute the minimal weak completion of a collection of
sets enabling exact representation of header classes of the set of rules
of a network. We believe that the incremental algorithms we have proposed
for updating such a representation may greatly be optimized, especially if
we make more assumptions about how sets of headers are represented. 
In the context of SDN for example, it is possible to use a data-structure 
for storing a collection of sets that allows to retrieve efficiently 
the sets intersecting a given query set. Storing our representation  in
such a data-structure would greatly improve the running time of our algorithms.
We reserve for future work the study of optimized incremental algorithm
for maintaining a weak completion.

In the context of forwarding table verification,
our representation could be used to improve existing practical implementations
such as proposed by Khurshid et al.~\cite{veriflow}. All the sets of our
representation are indeed included in their representation as they can
be expressed as combination of rules by intersection. However, they produce
many other sets (notably because of use of set differences in addition to
intersections) that are not really
necessary for representing header classes.

Weak completion could also be used as a replacement
of the strong completion proposed by Boutier and Chroboczek~\cite{boutier} for 
disambiguating a forwarding table with a minimal number of additional rules
in the SDN or firewall context. 
Strong completion consists in producing all possible combinations of the
original rules of the table and their number can quadratic in the number of
classes (i.e. the size of a minimal weak completion) in general as discussed
in \Cref{sub:weakcompl}.

\paragraph{Acknowledgment} The work presented in this paper has been partly carried out at LINCS (\url{http://www.lincs.fr}).

\newpage

\bibliographystyle{plain}

\newpage

\section*{Appendix A: proof of Theorem~\ref{th:repr}}

\begin{proof}[of \Cref{th:repr}]
We prove that the combination collection $\C$ computed by 
\Cref{algo:weak} provides a weak completion of $\R$.
Let $r_1,\ldots,r_n$ denote the rules in $\R$ in the same order as they 
are considered by the main for loop of \Cref{algo:weak}.
Set $\C_0=\set{H}$ and let $\C_i$ denote the combination collection contained
in $\C$ after having processed $r_1,\ldots,r_i$.
We show by induction on $i$ 
that $\C_i$ is a weak completion of $r_1,\ldots,r_i$.

Initially, $\C_0$ is clearly weakly complete. We thus consider $\C_i$ for $i>0$
assuming that $\C_{i-1}$ is a weak completion of $r_1,\ldots,r_{i-1}$.
We first show that $\C_i'=\C_i\cup\set{c\cap r_i\mid c\in \C_{i-1}}$ 
is weakly complete.
Considering $r,r'$ in that collection, we have to show the weakly complete
property $r\cap r'= \cup_{r''\subseteq r\cap r'}r''$.
This comes from the weak completeness of $\C_{i-1}$ if both $r$ and $r'$ where
already in $\C_{i-1}$. Otherwise, $r\cap r'=s\cap s'\cap r_i$ for some
$s,s'\in \C_{i-1}$. By weak completeness of $\C_{i-1}$,
we have $s\cap s' = \cup_{s''\subseteq s\cap s', s''\in \C_{i-1}}\ s''$ and thus
$s\cap s'\cap r_i = \cup_{s''\subseteq s\cap s', s''\in \C_{i-1}}\ s''\cap r_i$.
For all $s''\in \C_{i-1}$, $s''\cap r_i$ is in $\C_i'$ by
construction and is included in
$r\cap r' = s\cap s'\cap r_i$. We thus get 
$r\cap r'= \cup_{s''\subseteq r\cap r', s''\in \C_i'}\ s''$.
This proves that $\C_i'$ is weakly complete.
As $\C_i$ is obtained by removing sets which are not atom representative
from $\C_i'$, we have $\C_i=\AT(\C_i')$ and the weak completeness of
$\C_i$ follows from \Cref{cor:atomweak}.

We now show that $\C_i$ covers $r_1,\dots,r_i$. From the weak completeness
of $\C_{i-1}$, we have $H=\cup_{c\in \C_{i-1}}c$.
This implies $r_i=\cup_{c\in \C_{i-1}}c\cap r_i$. By definition of $\C_i'$
it thus covers $r_i$ in addition of $r_1,\ldots,r_{i-1}$ which are covered
by induction hypothesis. Since $\C_i=\AT(\C_i')$, \Cref{cor:atomcov}
implies that $\C_i$ also covers $r_1,\ldots,r_i$. It is thus a weak completion
of $r_1,\ldots,r_n$. And we get by induction that $\C$ is a weak completion of 
$\R$. 

This weak completion is minimal: any weak completion $\C'$ of $\R$
as same atoms as $\R$ by \Cref{lem:repr,lem:completion}.
Consider an atom $A$ of $\R$ and its atom representative sets $c,c'$
in $\C$ and $\C'$ respectively. 
By \Cref{lem:cover}, $\C$ covers $\C'$ and is thus
a representative collection of $\C'$ by \Cref{prop:weakrepr}.
Any set $s\in \C'$ thus contains $A$ when it contains $c$ and 
we thus have $c\subseteq c'$. We can symmetrically deduce $c'\subseteq c$.
As $\C$ contains only atom representative sets by construction, the collection
$\C$ is thus included in $\C'$. This proves that $\C$ is a minimal weak
completion and that such minimal weak completion is unique.

The computation time of \Cref{algo:weak} is dominated
by the computation of the cardinals of $a(r)$ for $r\in\C$ at each iteration
of the main for loop.
Its time complexity is thus a consequence of \Cref{prop:card}.
\end{proof}

\end{document}